\documentclass[11pt]{article}
\usepackage[margin=1in]{geometry}
\date{\vspace{-5ex}}
    \usepackage{cite}
	\usepackage{graphicx}
\usepackage{amsmath}
\usepackage{algorithm}
\usepackage{algpseudocode}

	\usepackage{amssymb}
	\usepackage{color}
	\usepackage[mathscr]{euscript}
	\usepackage{fixmath}
	\usepackage{url}
	\usepackage{tkz-graph}
	\usepackage{mdframed}
		 \usepackage{siunitx}
	 \usepackage{relsize}
	 \usepackage{subcaption}
	 \usepackage{bm}
	 \usepackage{siunitx}
	 \usepackage{wrapfig}
	 \usepackage{dsfont}
	 \usepackage{enumerate}
	 \usepackage{balance}
	 \usepackage{empheq}
	 \usepackage{mathtools}

\allowdisplaybreaks

\newtheorem{theorem}{Theorem}
\newtheorem{lemma}{Lemma}
\newtheorem{proposition}{Proposition}

\newtheorem{definition}{Definition} 
\newtheorem{proof}{Proof}

	\DeclareMathAlphabet\mathbfcal{OMS}{cmsy}{b}{n}

	\title{\LARGE \bf
	 Robust Distributed Control\\ Beyond Quadratic Invariance
}

\author{$^{~}$Luca Furieri and Maryam Kamgarpour
\thanks{This research was gratefully funded by the European Union ERC Starting Grant CONENE. $^{~}$The authors are with the Automatic Control Laboratory, D-ITET, ETH Zurich, Switzerland. e-mails: {\tt\footnotesize \{furieril, mkamgar\}@control.ee.ethz.ch}}%
}

\begin{document}

	\maketitle

	\begin{abstract}
	The problem of robust distributed control arises in several large-scale systems, such as transportation networks and power grid systems. In many practical scenarios controllers might not have enough information to make globally optimal decisions in a tractable way. We propose a novel class of tractable optimization problems whose solution is a controller complying with any specified information structure. The approach we suggest is based on decomposing intractable information constraints into two subspace constraints in the disturbance feedback domain.  We discuss how to perform the decomposition in an optimized way. The resulting control policy is globally optimal when a condition known as Quadratic Invariance (QI) holds, whereas it is feasible and it provides a provable upper bound on the minimum cost when QI does not hold.   Finally, we show that our method can lead to improved performance guarantees with respect to previous approaches, by applying the developed techniques to the platooning of autonomous vehicles. 

	
	\end{abstract}	
	\section{Introduction}
	


 Safe and efficient operation of large-scale systems such as the electric power grid, digital communication networks, autonomous vehicles, and financial systems,  relies on coordinating the decision making of multiple interacting agents. In most practical scenarios these agents can only base their decisions on partial local information due to geographic distance, privacy concerns or the high cost of communication. Lacking full information can make the task of designing optimal decisions significantly more challenging.

 The celebrated work of \cite{Witsenhausen} highlighted  that optimal control policies for the  Linear Quadratic Gaussian  problem given partial information may be nonlinear.   The intractability inherent to  lack of full information was  investigated in \cite{blondel2000survey, papadimitriou1986intractable}. The core challenges discussed in these works motivated  identifying special cases of optimal control problems with partial information for which efficient algorithms can be used. 
 Several cases of tractable problems with partial output information were characterized in \cite{ho1972team,bamieh2005convex}  and were later  generalized in \cite{rotkowitz2006characterization}, where the authors established  necessary and sufficient conditions  referred to as \emph{quadratic invariance} (QI) to obtain an exact convex reformulation.   The framework of QI was derived in the context of infinite horizon optimization, whereas several modern control architectures are based on finite horizon optimization.  In most applications, safety constraints must be robustly satisfied despite disturbances.  Motivated by these requirements, the works \cite{CDC2017,furieri2017value} showed that the QI notion can be exploited to include robust satisfaction of state and input constraints by establishing a connection between Youla parametrizations and disturbance feedback control policies. 
%


Nevertheless, the conditions posed by QI  can be too stringent for practical purposes. For instance, when the dynamics of the system evolve according to a strongly connected topology, delayed information about every output is required for each control input \cite{furieri2017value}. Sharing  this information might be limited by bandwidth and other network restrictions. Furthermore, controllers may be unable to share information due to strict privacy requirements.




Currently, it is not known how one could obtain an exact convex representation of state/output feedback distributed control problems given an information structure which is \emph{not} QI. For this reason, it is important to develop methods to compute approximate solutions. A line of work has focused on computing possibly suboptimal static distributed control policies using numerical methods \cite{fattahi2016theoretical,lin2011augmented,arastoo2014optimal}. 
Other approaches instead, hinge on exploiting the Youla parametrization to determine possibly suboptimal dynamic control policies.  Restricting the information structure to restore QI and obtain upper bounds on the minimum cost of the original problem was considered in \cite{rotkowitz2012nearest}. The work in \cite{matni2013heuristic} introduced the concept of optimization over \emph{QI covers} in the presence of delay constraints. It was shown that the iterative procedure proposed in \cite{matni2013heuristic} yields globally optimal solutions in certain cases. However, the assumption is that controllers have access to delayed measurements of all the outputs. With a similar philosophy to \cite{rotkowitz2012nearest,matni2013heuristic}, we aim at characterizing tractable formulations of the generally intractable distributed control problem, which are valid for any system and  information structure. When QI does not hold, our feasible solution provides a provable upper bound on the minimum cost which is at least as low as the one obtained with the techniques in \cite{rotkowitz2012nearest}. In general, the upper bound we derive can be strictly lower than \cite{rotkowitz2012nearest} and thus yields a more performing  distributed control policy.

Our main contributions are as follows. First, we suggest a novel class of  subspace constraints on the disturbance feedback parameter which preserve the given information structure.  Second, we determine optimized choices for the subspace constraints, with the goal of lowering these upper bounds.  Last, we show  that the developed techniques can lead to improved performance guarantees with respect to the approach of \cite{rotkowitz2012nearest} through a platooning example, arising in autonomous vehicles.

Section \ref{se:setup} sets up the problem. Section \ref{se:parametrizing} contains our main results about upper bounds on the minimum cost of control policies complying with any  specified information structure. Section~\ref{se:upperbounds} draws a connection with QI and  the approach in \cite{rotkowitz2012nearest}. The application to the platooning of vehicles is studied  in Section~\ref{se:application}.

\emph{Notation:~~}
		Given a matrix $Y \in \mathbb{R}^{a \times b}$ we refer to its scalar element located at row $i$ and column $j$ as $Y(i,j)$. The symbol $I_a$ denotes the identity matrix of dimensions $a \times a$ while $0_{a \times b}$ denotes the zero matrix of dimensions $a \times b$, for every $a \in \mathbb{Z}_{[1,\infty)}$ and $b \in \mathbb{Z}_{[1,\infty)}$. Given a binary matrix $X \in \{0,1\}^{a \times b}$ we define the subspace $\text{Sparse}\left(X\right) \subseteq \mathbb{R}^{a \times b}$ as $\{Y \in \mathbb{R}^{a \times b}|~Y(i,j)=0,~ \forall i,j \text{ s.t. } X(i,j)=0\}$.
	Given $Y \in \mathbb{R}^{a \times b}$ we define $X=\text{Struct}(Y)$ as the binary matrix
	\begin{equation*}
	X(i,j)=\begin{cases}
	1\quad \text{if }Y(i,j)\neq 0\,,\\
	0 \quad \text{otherwise}\,.
	\end{cases}
	\end{equation*}
	 Let $X,X' \in \{0,1\}^{a \times b}$ be binary matrices. Throughout the paper we will adopt the following conventions. 	  $XX':=\text{Struct}(XX')$ and $X^r:=\text{Struct}(X^r)$. $X \leq X'$ if and only if $X(i,j)\leq X'(i,j)~\forall i,j$.  $X < X'$ if and only if $X \leq X'$ and there exist indices $i,j$ such that $X(i,j)<X'(i,j)$.  $X \not\leq X'$ if and only if there exist indices $i,j$ such that $X(i,j)>X'(i,j)$.

\section{Problem Formulation}
\label{se:setup}

	We consider a discrete time system
	\begin{equation}
	\label{eq:system}
	\begin{aligned}
	&x_{k+1}=Ax_k+Bu_k+Dw_k\,, \quad y_k=Cx_k+Hw_k\,,\\
	\end{aligned}
	\end{equation}
	where $k \in \mathbb{Z}_{[0,\infty)}$, $x_k \in \mathbb{R}^n$, $y_k \in \mathbb{R}^p$, $u_k \in \mathbb{R}^m$, $w_k\in \mathcal{W}$ and $\mathcal{W} \subseteq \mathbb{R}^n$ is the set of possible disturbances.  The system starts from a known initial condition $x_0\in \mathbb{R}^n$. Let us define a prediction horizon of length $N$. Our goal is to minimize a cost function of the history of states and inputs $J(x_0,\cdots,x_N,u_0,\cdots,u_{N-1})$. Furthermore, the states and inputs need to satisfy
	\begin{equation}
	\label{eq:constraints_state_input}
	\begin{aligned}
	&\begin{bmatrix}x_k\\u_k\end{bmatrix} \in \Gamma\subseteq \mathbb{R}^{n+m}\,,\quad x_N \in \mathcal{X}_f \subseteq \mathbb{R}^{n}\,,
	\end{aligned}
	\end{equation}
	 for all $k \in \mathbb{Z}_{[0,N-1]}$ and for all realizations of disturbances taken from set $\mathcal{W}$. Each control input can depend on the history of a subset of outputs. This subset is described by the so-called 	\emph{information structure}.  An information structure can be time varying, in the sense that controllers may measure, memorize or forget specific outputs at different times.
	



	The search in the class of all output feedback policies is  intractable.   Hence, a possible approach is to restrict the search to the class of controllers that are affine in the history of the outputs. The time varying output feedback affine  control policy is expressed as
		\begin{equation}
	\label{eq:affine_feedback}
	\begin{aligned}
	&u_k=\sum_{j=0}^kL_{k,j}y_j+g_k\,,
	\end{aligned}
	\end{equation}
	for all time instants $k \in \mathbb{Z}_{[0,N-1]}$.
	
	 For every  $j \in \mathbb{Z}_{[0,k]}$, we consider binary matrices $S_{k,j}\in \{0,1\}^{m \times p}$ encoding what controllers at time $k$ know about the outputs at time $j\leq k$, that is $S_{k,j}(a,b)=1$ if and only if controller $a$ at time $k$ knows the $b$-th output  at time $j$. Let $\mathcal{S}_{k,j} \subseteq \mathbb{R}^{m \times p}$  denote the \emph{sparsity subspace} generated by the binary matrix $S_{k,j}$ in the sense that $\mathcal{S}_{k,j}=\text{Sparse}(S_{k,j})$.   The information structure on the input can thus be equivalently formulated as 
	\begin{equation}
	\label{eq:affine_feedback_constraints}
	\begin{aligned}
	&L_{k,j} \in \mathcal{S}_{k,j},~ g_k \in \mathbb{R}^{m}\,,\\
	\end{aligned}\
	\end{equation}
	for every $k \in \mathbb{Z}_{[0,N-1]}$ and  $j \in \mathbb{Z}_{[0,k]}$.  
We are now ready to state the optimization problem under study.
	\begin{empheq}[box=\fbox]{alignat*=3}
	&\text{\textbf{Problem 1}}&&~\\
	 &\min&&J(x_0,\cdots,x_N,u_0,\cdots,u_{N-1})\\ 
	 &~\text{s.t. }&&(\ref{eq:system}),(\ref{eq:constraints_state_input}) \quad \forall w_k\in \mathcal{W}\,,\\
	&~&& (\ref{eq:affine_feedback}),(\ref{eq:affine_feedback_constraints}) \quad  \forall k \in \mathbb{Z}_{[0,N-1]}\,, \forall j \in \mathbb{Z}_{[0,k]}.
		 	\end{empheq}
	In the above, the decision variables are the matrices $L_{k,j}$ and the vectors $g_k$  as in (\ref{eq:affine_feedback_constraints}) for all $k \in \mathbb{Z}_{[0,N-1]}$ and $j \in \mathbb{Z}_{[0,k]}$. For computational tractability we assume that $J(\cdot)$ is a convex function of the disturbance-free state and input trajectories (that is, when the disturbances are assumed not to be present)  and that the sets $\Gamma,~\mathcal{X}_f$ are polytopes  $\Gamma=\{(x,u) \in (\mathbb{R}^{n}, \mathbb{R}^m)\text{ s.t. }Ux+Vu\leq b\}$, where $U \in \mathbb{R}^{s\times n}$, $V \in \mathbb{R}^{s \times m}$ and $b \in \mathbb{R}^s$, and  $\mathcal{X}_f=\{x \in \mathbb{R}^{n}\text{ s.t. }Rx \leq z\}$,
where $R \in \mathbb{R}^{r \times n}$ and	$z \in \mathbb{R}^r$.   It is convenient to define the vectors of stacked variables as 
	\begin{equation*}
	\begin{aligned}
	&\mathbf{x}=\begin{bmatrix}x_0^\mathsf{T}&\cdots&x_N^\mathsf{T}\end{bmatrix}^{\mathsf{T}} \in \mathbb{R}^{n(N+1)}\,,\\
	&\mathbf{y}=\begin{bmatrix}y_0^\mathsf{T}&\cdots&y_{N}^\mathsf{T}\end{bmatrix}^\mathsf{T} \in \mathbb{R}^{p(N+1)}\,,\\
	 &\mathbf{u}=\begin{bmatrix}u_0^\mathsf{T}&\cdots&u_{N-1}^\mathsf{T}&0_{m \times 1}^\mathsf{T}\end{bmatrix}^\mathsf{T}\in \mathbb{R}^{m(N+1)}\,,\\
	&\mathbf{w}=\begin{bmatrix}w_0^\mathsf{T}&\cdots&w_{N-1}^{\mathsf{T}}&0_{n \times 1}^\mathsf{T}\end{bmatrix}^{\mathsf{T}}\in \mathbb{R}^{n(N+1)}\,.\\
	\end{aligned}
	\end{equation*}
	Equation (\ref{eq:system}) can be succinctly expressed as
	\begin{equation}
	\label{eq:states_stacked}
	\begin{aligned}
	&\mathbf{x}=\mathbf{A}x_0+\mathbf{Bu}+\mathbf{E}_D\mathbf{w}\,, \quad \mathbf{y}=\mathbf{Cx+Hw}\,,\\
	\end{aligned}\
	\end{equation}
where matrices $\mathbf{A}$, $\mathbf{B}$, $\mathbf{E}_D$, $\mathbf{C}$ and $\mathbf{H}$ are defined in the Appendix. Their derivation is straightforward from the recursive application of (\ref{eq:system}).
	Similarly,  considering (\ref{eq:affine_feedback}), the control input can be expressed as	$\mathbf{u}=\mathbf{Ly+g}$,	where $\mathbf{L}\in \mathbb{R}^{m(N+1)\times p(N+1)}$ and $\mathbf{g}\in \mathbb{R}^{m(N+1)}$ are defined in the  Appendix.  In order to satisfy (\ref{eq:affine_feedback_constraints}) matrix $\mathbf{L}$ must lie in the subspace $\mathbfcal{S} \subseteq \mathbb{R}^{m(N+1) \times p(N+1)}$, where $\mathbfcal{S}=\text{Sparse}(\mathbf{S})$ and $\mathbf{S}$ stacks the matrices $S_{k,j}$'s as per (\ref{eq:L}) in the Appendix.	

	 Problem~1 is non-convex regardless of the sparsity constraints being linear in $\mathbf{L}$. It is known that when an information structure is not given, parametrizing the controller as a disturbance feedback affine policy restores tractability of Problem~1 \cite{Goulart}.  Letting $\mathbf{P=C}\mathbf{E}_D+\mathbf{H}$ we define the disturbance feedback controller as $\mathbf{u}=\mathbf{QPw}+\mathbf{v}$.
 The decision variable  $\mathbf{Q}\in \mathbb{R}^{m(N+1)\times p(N+1)}$  is causal as in (\ref{eq:L}).
	It is possible to map a disturbance feedback controller $(\mathbf{Q,v})$ to the unique corresponding output feedback controller $(\mathbf{L,g})$ and vice versa as follows.
\begin{align}
&\mathbf{L}=\mathbf{Q}(\mathbf{CBQ}+I_{p(N+1)})^{-1}\,,\label{eq:map_Q_to_L}\\
&\mathbf{g}=\mathbf{v}-\mathbf{Q}(\mathbf{CBQ}+I_{p(N+1)})^{-1}(\mathbf{CBv+CA}x_0)\,,\nonumber\\
&\mathbf{Q}=\mathbf{L}(I_{p(N+1)}-\mathbf{CBL})^{-1}\,, \label{eq:map_L_to_Q}\\
&\mathbf{v}=\mathbf{L}(I_{p(N+1)}-\mathbf{CBL})^{-1}(\mathbf{CBg+CA}x_0)+\mathbf{g}\,.\nonumber
\end{align}
It is easy to show that when $J(\cdot)$ is convex in the disturbance-free states and inputs it is also convex in $(\mathbf{Q},\mathbf{v})$. The state and input constraints (\ref{eq:constraints_state_input}) are linear in ($\mathbf{Q,v}$) and  can be expressed as
\begin{equation*}
F\mathbf{v}+\max_{\mathbf{w}\in \mathcal{W}^{N +1}}(F\mathbf{QP}+G)\mathbf{w}\leq c\,,
\end{equation*}
 where $F\in \mathbb{R}^{(Ns+r) \times m(N+1)},~G \in \mathbb{R}^{(Ns+r) \times n(N+1)},~c \in \mathbb{R}^{Ns+r}$ are reported in the Appendix for completeness. However,  the sparsity constraint $\mathbf{L}=\mathbf{Q}(\mathbf{CBQ}+I_{p(N+1)})^{-1} \in \mathbfcal{S}$ is nonlinear in $\mathbf{Q}$ in general. Thus, we consider the following convex program:
 \begin{empheq}[box=\fbox]{alignat*=3}
	&\textbf{Problem~2}\\
	 &\min_{\mathbf{Q,v}}&&J(x_0,\mathbf{Q},\mathbf{v})\\ 
	 &~\text{s.t. }&&\mathbf{Q}\text{ is causal}\,,\\
	 &~&&F\mathbf{v}+\max_{\mathbf{w}\in \mathcal{W}^{N +1}}(F\mathbf{QP}+G)\mathbf{w}\leq c\,,\\
	 &~&& \mathbf{Q} \in \mathbfcal{R}\,,
	 \end{empheq}
	 where $\mathbfcal{R}$ is a subspace that must be designed to preserve the sparsity of $\mathbf{L}$ through the mapping (\ref{eq:map_Q_to_L}). To simplify the notation we introduce the following definition.
\begin{definition}
\emph{Let $X \in \mathbb{R}^{a \times b}$ and $Y \in \mathbb{R}^{b\times a}$. The closed-loop function  $h:\mathbb{R}^{a\times b}\times \mathbb{R}^{b\times a}\longrightarrow ~\mathbb{R}^{a \times b}$ is defined as $h(X,Y)=-X(I_b-YX)^{-1}$. For a set $\mathcal{X} \subseteq \mathbb{R}^{a\times b}$ the set $h(\mathcal{X},Y)$ is defined as $h(\mathcal{X},Y)=\{h(X,Y),~X\in \mathcal{X}\}$.}
\end{definition}

Notice that the operator $h(\cdot)$ maps a disturbance feedback controller $\mathbf{Q}$ to a corresponding output feedback controller $\mathbf{L}$. In particular, the mappings (\ref{eq:map_Q_to_L}) and (\ref{eq:map_L_to_Q}) can be expressed as $\mathbf{L}=h(\mathbf{-Q},\mathbf{CB})$ and $\mathbf{Q}=-h(\mathbf{L},\mathbf{CB})$ respectively. Accordingly, to preserve the sparsity of $\mathbf{L}$ through mapping (\ref{eq:map_Q_to_L}) we require that $\mathbfcal{R}$ is designed such that
 \begin{equation}
\label{eq:generalized_set_requirement}
h(\mathbfcal{R},\mathbf{CB}) \subseteq \mathbfcal{S}\,,
\end{equation}
where we used the fact that $-\mathbfcal{R}=\mathbfcal{R}$ because $\mathbfcal{R}$ is a subspace. Whenever $\mathbfcal{R}$ satisfies (\ref{eq:generalized_set_requirement}) we equivalently say that it is \emph{sparsity preserving}. Refer to Figure~\ref{fig:mapping} to visualize (\ref{eq:generalized_set_requirement}). Notice that $h(\mathbfcal{R},\mathbf{CB})$ is a non-convex set in general despite $\mathbfcal{R}$ being convex, because $h(\cdot,\mathbf{CB})$ is a nonlinear map.

We remark that designing $\mathbfcal{R}$ is not a trivial task. For instance, if we simply require $\mathbf{Q} \in \mathbfcal{S}$ the resulting $\mathbf{L}$ might not lie in $\mathbfcal{S}$.  Therefore,  in the next section we search for subspaces $\mathbfcal{R}$ so as to formulate a tractable convex program (Problem~2) whose  feasible solutions correspond to feasible $\mathbf{L}$'s for the original problem (Problem~1). The optimal solution of Problem~2 will thus correspond to a provable upper bound on the minimum cost of Problem~1.
	\begin{figure}[thb]
	      \centering
	      	\includegraphics[width=0.5\textwidth]{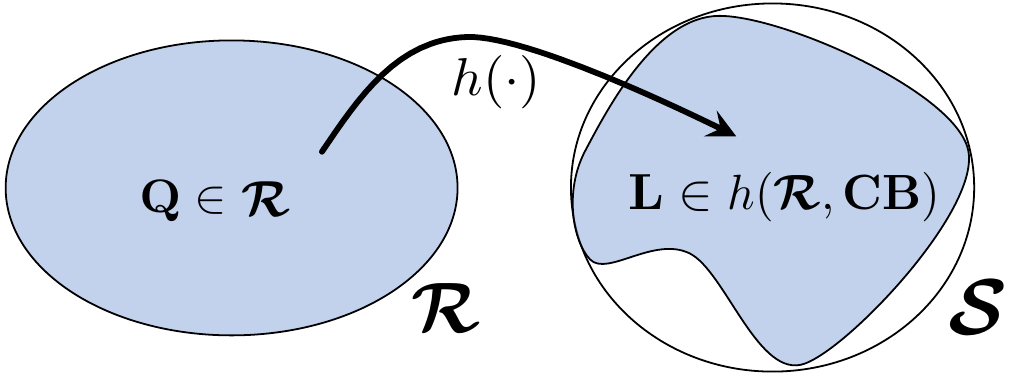}
	      		      	\setlength{\belowcaptionskip}{-17pt}

	      	\caption{ \small Sparsity preserving subspaces}
	      		\label{fig:mapping}
	\end{figure} 
\section{ Solution Approach}
\label{se:parametrizing}
 We  describe the key idea leading to the construction of the proposed family of sparsity preserving subspaces as in (\ref{eq:generalized_set_requirement}). Consider equation (\ref{eq:map_Q_to_L}) and notice that $\mathbf{L}$ can be equivalently written as $\mathbf{L}=(I_{m(N+1)}+\mathbf{QCB})^{-1}\mathbf{Q}$. By taking the power series expansion of the inverse matrix and exploiting the fact that  $\mathbf{QCB}$ is nilpotent because its diagonal is zero by construction, we obtain
\begin{equation}
\label{eq:power_expansion}
\mathbf{L}=\sum_{i=0}^{N-1}(\mathbf{QCB})^i\mathbf{Q}\,.
\end{equation}
Our approach is to ensure that every addend in (\ref{eq:power_expansion}) lies in $\mathbfcal{S}$ by means of subspace constraints on $\mathbf{Q}$. This can be done by decomposing the term $\mathbf{QCBQ}$, which is quadratic in $\mathbf{Q}$, into two  factors $\mathbf{QCB}$ and $\mathbf{Q}$, which are both linear in $\mathbf{Q}$. Both linear terms can be required to have certain sparsity patterns so that all the addends of (\ref{eq:power_expansion}) in the form $(\mathbf{QCB})^i\mathbf{Q}$ lie in $\mathbfcal{S}$ for all $i \in \mathbb{Z}_{[0,N-1]}$. This key idea leads to defining the notion of \emph{Generalized Sparsity Subspaces} (GSS) as follows.
\begin{definition}
\label{de:GSS}
\emph{For a given $G \in \mathbb{R}^{b \times a}$ and any matrices $T \in \{0,1\}^{a \times b}$, $Y\in \{0,1\}^{a \times a}$ the Generalized Sparsity Subspace (GSS)  $\mathbfcal{R}_G(T,Y) \subseteq \mathbb{R}^{a \times b}$ is defined as 
\begin{equation*}
\mathbfcal{R}_G(T,Y) =\{Q \in \text{Sparse}(T)~\text{ s.t. }~QG \in \text{Sparse}(Y)\}\,.
\end{equation*}}
\end{definition}

%
We say that these sparsity subspaces  are \emph{generalized} because the constraint that $\mathbf{Q}$ lies in $\mathbfcal{S}$, as commonly considered in the literature \cite{rotkowitz2006characterization,rotkowitz2012nearest,CDC2017},  can be obtained from a specific choice of $T,Y$ as above. In particular, by letting $T=\mathbf{S}$ and $Y\geq\mathbf{S\Delta}$, where $\mathbf{\Delta}=\text{Struct}(\mathbf{CB})$, we have that $\mathbfcal{R}_{\mathbf{CB}}(T,Y)=\mathbfcal{S}$. This is because if $\mathbf{Q} \in \mathbfcal{S}$ then $\mathbf{QCB} \in \text{Sparse}(\mathbf{S\Delta})$ by construction (see Lemma~\ref{le:structure_multiplication} ahead). 

Next, we derive our main result on the conditions ensuring that a GSS is sparsity preserving. 
\begin{theorem}
\label{th:GSPS}
Let $\mathbf{T,Y}$ be binary matrices such that $\mathbf{T}\leq \mathbf{S}$ and $\mathbf{YT}\leq \mathbf{T}$. Then, the subspace $\mathbfcal{R}_{\mathbf{CB}}(\mathbf{T,Y})$ is sparsity preserving as per (\ref{eq:generalized_set_requirement}).
\end{theorem}

	The proof of Theorem~\ref{th:GSPS} relies on the following  two lemmas.
		 \begin{lemma}
	 \label{le:structure_multiplication}
 Let $X_1\in \{0,1\}^{a \times b}$, $X_2 \in \{0,1\}^{b \times a}$ and let $\mathcal{X}_1=\text{Sparse}(X_1)$, $\mathcal{X}_2=\text{Sparse}(X_2)$. Let $Q_1 \in \mathcal{X}_1$ and $Q_2 \in \mathcal{X}_2$. Then $Q_1Q_2 \in \text{Sparse}(X_1X_2)\,.$
	 \end{lemma}
	 	 \begin{proof}
\emph{Suppose  $Q_1Q_2(i,j)\neq 0$. Then $\sum_{k=1}^b Q_1(i,k)Q_2(k,j)\neq 0$ and so there is $\bar{k}$ such that $Q_1(i,\bar{k})\neq 0$ and $Q_2(\bar{k},j) \neq 0$. Since $Q_1 \in \mathcal{X}_1$ and $Q_2 \in \mathcal{X}_2$, then $X_1(i,\bar{k})=1$ and $X_2(\bar{k},j)=1$. Hence, $X_1X_2(i,j)=1$. The same reasoning holds for all indices $i,j$ such that  $Q_1Q_2(i,j)\neq 0$. Hence, $Q_1Q_2 \in \text{Sparse}(X_1X_2)\,.$}
	 \end{proof}
	 	\begin{lemma}
	\label{le:powers_Y}
Let $Y,T$ be binary matrices of compatible dimensions such that $YT\leq T$. Then $Y^iT \leq T$ for every positive integer $i$.
	\end{lemma}
\begin{proof}
\emph{We prove that  if $YT\leq T$ and  $X\leq T$, then $YX \leq T$.  For any $i,j$ such that $YX(i,j)=1$,  there exists $k$ such that $Y(i,k)=X(k,j)=1$. Then, $T(k,j)=1$ because $X\leq T$ and $YT(i,j)=1$ because $Y(i,k)=T(k,j)=1$. This implies that $YX\leq YT$. We know by hypothesis that $YT\leq T$, hence $YX\leq YT \leq T$. Next, we prove the statement by induction. The base case $YT\leq T$ holds by hypothesis. Suppose that $Y^{i-1}T \leq T$. Then $Y^iT=Y(Y^{i-1}T) \leq T$ follows by $Y^{i-1}T \leq T$ and the observations above.}  
	\end{proof}
	We are now ready to prove Theorem~\ref{th:GSPS}.	
	
	\begin{proof}(Proof of Theorem~\ref{th:GSPS})
 \emph{Let $\mathbf{T}\leq \mathbf{S}$ and $\mathbf{YT}\leq \mathbf{T}$, and take a generic $\mathbf{Q} \in \mathbfcal{R}_{\mathbf{CB}}(\mathbf{T,Y})$. By definition, $\mathbf{Q} \in \text{Sparse}(\mathbf{T})$ and $\mathbf{QCB} \in \text{Sparse}(\mathbf{Y})$. Consider the corresponding $\mathbf{L}=\sum_{i=0}^{N-1}(\mathbf{QCB})^{i}\mathbf{Q}$. The addend corresponding to $i=0$ is $\mathbf{Q}$ and clearly lies in $\mathbfcal{S}$, because $\mathbf{T}\leq \mathbf{S}$ and thus $\mathbf{Q} \in \text{Sparse}(\mathbf{T}) \subseteq \mathbfcal{S}$. By Lemma~\ref{le:powers_Y}, $\mathbf{Y}^i\mathbf{T}\leq \mathbf{T}$ for every positive integer $i$ and thus $\text{Sparse}(\mathbf{Y}^i\mathbf{T})\subseteq \text{Sparse}(\mathbf{T})$.  Lemma~\ref{le:structure_multiplication} states that $(\mathbf{QCB})^i \mathbf{Q} \in \text{Sparse}(\mathbf{Y}^i\mathbf{T})$ for any positive integer $i$. Since $\text{Sparse}(\mathbf{Y}^i\mathbf{T})\subseteq \text{Sparse}(\mathbf{T})$, then $(\mathbf{QCB})^i \mathbf{Q} \in \text{Sparse}(\mathbf{T})$. Hence, every addend in the power series expansion of $\mathbf{L}$ lies in $\text{Sparse}(\mathbf{T}) \subseteq \mathbfcal{S}$ and we conclude $\mathbf{L} \in \mathbfcal{S}$. This concludes the proof because $\mathbf{Q}$ was a generic element of $\mathbfcal{R}_{\mathbf{CB}}(\mathbf{T,Y})$.}
	\end{proof}



We remark that the conditions of Theorem~\ref{th:GSPS} are sufficient to satisfy (\ref{eq:generalized_set_requirement}), but not necessary. This is because $\mathbf{L}$ as in (\ref{eq:power_expansion}) can lie in $\mathbfcal{S}$ even if some of the addends do not lie in $\mathbfcal{S}$.

\subsection{Lowering the upper bounds}
\label{sub:lowering}

To obtain the least upper bound on the minimum cost of Problem~1 we need to maximize the set $h(\mathbfcal{R}_{\mathbf{CB}}(\mathbf{T,Y}),\mathbf{CB})$ with the constraint that it  remains a subset of $\mathbfcal{S}$. Visually, this is equivalent to enlarging the (possibly non-convex) set on the right side of Figure~\ref{fig:mapping} so that it still fits inside $\mathbfcal{S}$. 

Suppose that  $\mathbf{T}\leq \mathbf{S}$ is fixed. Using Theorem~\ref{th:GSPS}, the maximization problem described above is equivalent to designing $\mathbf{Y}_{\text{max}}^{\mathbf{T}}$ such that $\mathbf{Y}_{\text{max}}^{\mathbf{T}}\mathbf{T}\leq \mathbf{T}$ and for any other $\mathbf{Y}$ with $\mathbf{YT} \leq \mathbf{T}$, the following condition holds
\begin{equation}
\label{eq:optimal_set_requirement}
 h\left(\mathbfcal{R}_{\mathbf{CB}}(\mathbf{T,Y}),\mathbf{CB}\right) \subseteq h\left(\mathbfcal{R}_{\mathbf{CB}}(\mathbf{T,\mathbf{Y}}_{\text{max}}^{\mathbf{T}}),\mathbf{CB}\right)\,.
\end{equation}
The set inclusion (\ref{eq:optimal_set_requirement}) and the above requirements on $\mathbf{Y}_{\text{max}}^{\mathbf{T}}$ can be restated as follows
\begin{equation}
\label{eq:requirement_binary}
1)\mathbf{Y}_{\text{max}}^{\mathbf{T}}\mathbf{T}\leq \mathbf{T},~~\text{ and }~~ 2)\mathbf{YT} \leq \mathbf{T} \implies \mathbf{Y} \leq \mathbf{Y}_{\text{max}}^{\mathbf{T}} \,.
\end{equation}
This is because for any  function $f:\mathcal{D}\rightarrow \mathcal{C}$ and sets $\mathcal{X}_1 \subseteq \mathcal{X}_2 \subseteq \mathcal{D}$  we have that $f(\mathcal{X}_1)\subseteq f(\mathcal{X}_2)$. Hence, the set inclusion (\ref{eq:optimal_set_requirement})   is implied by $\mathbfcal{R}_{\mathbf{CB}}(\mathbf{T,Y})\subseteq \mathbfcal{R}_{\mathbf{CB}}(\mathbf{T,\mathbf{Y}}_{\text{max}}^{\mathbf{T}})$. 

Using Definition~\ref{de:GSS} of a GSS, observe that condition (\ref{eq:requirement_binary}) amounts to requiring that the number of entries of the term $\mathbf{QCB}$ which are set to zero is kept to a minimum, while still ensuring that the sparsity of $\mathbf{L}$ is preserved. Matrix $\mathbf{Y}_{\text{max}}^{\mathbf{T}}$ as in (\ref{eq:requirement_binary}) is computed with a simple procedure described in the following proposition.
 \begin{proposition}
 \label{pr:cannot_be_larger}
Fix $\mathbf{T}\leq \mathbf{S}$.  Let $\mathbf{Y}_{\emph{\text{max}}}^{\mathbf{T}}(i,j)=0$ if  $k$ exists such that $\mathbf{T}(i,k)=0$ and $\mathbf{T}(j,k)=1$, and $\mathbf{Y}_{\emph{\text{max}}}^{\mathbf{T}}(i,j)=1$ otherwise. Then $\mathbf{Y}_{\emph{\text{max}}}^{\mathbf{T}}$ is maximal as per  (\ref{eq:requirement_binary}). 
 \end{proposition}
 \begin{proof}
 \emph{ Take indices $i,k$ such that  $\mathbf{T}(i,k)=0$. Then, $\mathbf{Y}_{\text{max}}^{\mathbf{T}}(i,j)=0$ for every $j\neq i$ such that $\mathbf{T}(j,k)=1$ by construction. This implies that $\mathbf{Y}_{\text{max}}^{\mathbf{T}}\mathbf{T}(i,k)={\mathlarger{\mathlarger{\mathlarger{\mathlarger{\lor}}}}}_{j=1}^{m(N+1)}\mathbf{Y}_{\text{max}}^{\mathbf{T}}(i,j)\mathbf{T}(j,k)=0$, where  ``$\lor$'' indicates the logical ``or'' operator.  Hence,  $\mathbf{Y}_{\text{max}}^{\mathbf{T}}\mathbf{T}\leq \mathbf{T}$. For the second of (\ref{eq:requirement_binary}) we reason by contrapositive. Suppose that $\mathbf{Y} \not \leq \mathbf{Y}_{\text{max}}^{\mathbf{T}}$. Then, there are indices $i,j$ such that $\mathbf{Y}(i,j)=1$ and $\mathbf{Y}_{\text{max}}^{\mathbf{T}}(i,j)=0$. By construction  $\mathbf{Y}_{\text{max}}^{\mathbf{T}}(i,j)=0$ implies that there exists index $k$ such that $\mathbf{T}(j,k)=1$ and $\mathbf{T}(i,k)=0$. As a consequence, $\mathbf{YT}(i,k)={\mathlarger{\mathlarger{\mathlarger{\mathlarger{\lor}}}}}_{r=1}^{m(N+1)}\mathbf{Y}(i,r)\mathbf{T}(r,k)=\mathbf{Y}(i,j)\mathbf{T}(j,k)=1$, where  ``$\lor$'' indicates the logical ``or'' operator. This implies that $\mathbf{YT}\not \leq \mathbf{T}$ because $\mathbf{T}(i,k)=0$.}
 \end{proof}

  For the rest of the paper $\mathbf{Y}_{\text{max}}^{\mathbf{T}}$  will always refer to $\mathbf{Y}$ being designed according to Proposition~\ref{pr:cannot_be_larger}. Summing up, the choice $\mathbf{Y}=\mathbf{Y}_{\text{max}}^{\mathbf{T}}$ ensures that $\mathbfcal{R}_{\mathbf{CB}}(\mathbf{T},\mathbf{Y}_{\text{max}}^{\mathbf{T}})$ is sparsity preserving  and maximal  as per (\ref{eq:optimal_set_requirement}). Hence, lowered upper bounds on the minimum cost of Problem~1  can be found by solving Problem~2 with $\mathbfcal{R}_{\mathbf{CB}}(\mathbf{T,Y_{\text{max}}^{\mathbf{T}}})$. Furthermore, notice that the design procedure in Proposition~\ref{pr:cannot_be_larger}  has a complexity of $O(ab^2)$, where $a$ and $b$ are the number of rows and columns of $\mathbf{T}$ respectively, and $\mathbf{Y}_{\text{max}}^\mathbf{T}$ can thus be computed in polynomial time.

At this point, you may wonder why we did not restrict our attention to the case where $\mathbf{T}$ is maximal, that is $\mathbf{T}=\mathbf{S}$. The reason is that the definition of the GSS  $\mathbfcal{R}_{\mathbf{CB}}(\cdot)$ depends on the system $\mathbf{CB}$. Because of this fact it is possible that $\mathbfcal{R}_{\mathbf{CB}}(\mathbf{S},\mathbf{Y}_{\text{max}}^{\mathbf{S}})$ is a strict subset of $\mathbfcal{R}_{\mathbf{CB}}(\mathbf{T},\mathbf{Y}_{\text{max}}^{\mathbf{T}})$ for some non-maximal $\mathbf{T}<\mathbf{S}$. 

\section{Upper Bounds Beyond Quadratic Invariance}
\label{se:upperbounds}

We start by recalling the notion of Quadratic Invariance.
\begin{definition}[Quadratic Invariance (QI)]
\label{de:QI}
\emph{Let  ${\mathcal{K} \subseteq \mathbb{R}^{a \times b}}$ be a subspace and $G\in \mathbb{R}^{b\times a}$.  Then $\mathcal{K}$ is QI with respect to $G$ if $KGK \in \mathcal{K}$ for every $K \in \mathcal{K}$.}
\end{definition}
As outlined,  our approach relies on some technical restrictions to guarantee that $\mathbf{L} \in \mathbfcal{S}$. You may then wonder how the suboptimality intrinsic to such restrictions compares with the previously proposed approach of finding  a QI sparsity subspace which is a close subset of $\mathbfcal{S}$ \cite{rotkowitz2012nearest}.
  \begin{proposition}
Let $\text{Sparse}(\mathbf{T}_{\text{QI}})\subseteq \mathbfcal{S}$ be QI with respect to $\mathbf{CB}$.   Then, there exists $\mathbf{T}\leq \mathbf{S}$ such that $J^\star \leq J_{\text{QI}}$, where $J^\star$ and $J_{\text{QI}}$ are the minimum costs of Problem~2 with  $\mathbfcal{R}=\mathbfcal{R}_{\mathbf{CB}}(\mathbf{T},\mathbf{Y}_{\text{max}}^{\mathbf{T}})$ and $\mathbfcal{R}=\text{Sparse}(\mathbf{T}_{\text{QI}})$ respectively.
  \end{proposition}
  \begin{proof}
   \emph{ We prove that since $\text{Sparse}(\mathbf{T}_{\text{QI}})$ is QI with respect to $\mathbf{CB}$ then $\mathbfcal{R}_{\mathbf{CB}}(\mathbf{T_{\text{QI}},T_{\text{QI}}\Delta})=\text{Sparse}(\mathbf{T}_{\text{QI}})$ where $\mathbf{\Delta}=\text{Struct}(\mathbf{CB})$. This is because by \cite[Theorem 26]{rotkowitz2006characterization} QI is equivalent to $\mathbf{T}_{\text{QI}}\mathbf{\Delta}\mathbf{T}_{\text{QI}} \leq \mathbf{T}_{\text{QI}}$. Then, by Proposition~\ref{pr:cannot_be_larger} and (\ref{eq:requirement_binary}) we obtain that $ \mathbf{T}_{\text{QI}}\mathbf{\Delta}\leq \mathbf{Y}_{\text{max}}^{\mathbf{T}_{\text{QI}}}$. Also observe that if $\mathbf{Q}\in \text{Sparse}(\mathbf{T}_{\text{QI}})$ then $\mathbf{QCB} \in  \text{Sparse}(\mathbf{Y}_{\text{max}}^{\mathbf{T}_{\text{QI}}})$ by the above and Lemma~\ref{le:structure_multiplication}. Then $\mathbfcal{R}_{\mathbf{CB}}(\mathbf{T},\mathbf{Y}_{\text{max}}^{\mathbf{T}})=\text{Sparse}(\mathbf{T}_{\text{QI}})$ by selecting $\mathbf{T}=\mathbf{T}_{\text{QI}}$. Hence, $J^\star = J_{\text{QI}}$ and $J^\star \leq J_{\text{QI}}$.}
  \end{proof}
The proposition above shows that the performance achieved by requiring that $\mathbf{Q}$ lies in any QI sparsity subspace  can always be matched by requiring that $\mathbf{Q}$ lies in a corresponding GSS. Our approach thus recovers the global optimum of Problem~1 if $\mathbfcal{S}$ is QI \cite{rotkowitz2006characterization} by solving Problem~2 with the GSS $\mathbfcal{R}_{\mathbf{CB}}(\mathbf{S},\mathbf{Y}_{\text{max}}^{\mathbf{S}})$. Furthermore, when $\mathbfcal{S}$ is not QI, we show with an application example that we can in general  select a GSS which \emph{outperforms} QI sparsity subsets of $\mathbfcal{S}$ as per \cite{rotkowitz2012nearest}. 
\vspace{-0.2cm}
\section{Application to Platooning}
\label{se:application}
The problem of platooning considers control of a set of vehicles moving on a straight line in presence of disturbances. The difficulty in this control problem lies in the fact that each vehicle has local information, based on measurements of the formation leader and/or the preceding vehicle only.  A review of the most prominent challenges and past solutions can be found in \cite{sabuau2017optimal,zheng2017distributed,sadraddini2017provably}. The problem of distributed set invariance applied to platooning was considered in \cite{sadraddini2018distributed}. Our goal here is to show applicability and efficacy of our approach to controller design when a \emph{predecessor-follower} \cite{sabuau2017optimal} non-QI information structure must be complied with.

We consider the platooning of $n$ vehicles modeled as point masses of mass $m$. For each vehicle, the engine thrust is modeled as a force control input acting on the point mass. The system in continuous time is expressed by equations
		\begin{equation*}
		\begin{aligned}
		&\dot{x}(t)=A_cx(t)+B_cu(t)+w(t)\,,\quad y(t)=Cx(t)+w(t)\,,
		\end{aligned}
		\end{equation*}
	where $x(t) \in \mathbb{R}^{2n}$ (the first $n$ state variables are the positions, the last $n$ are the velocities), $u(t) \in \mathbb{R}^n$ and $w(t) \in \mathcal{W}$ for every $t \in \mathbb{R}$,  $\mathcal{W} \subset \mathbb{R}^{2}$ being a polytope of disturbances for positions and velocities. Matrices $A_c$ and $B_c$ are written as
	\begin{equation*}
	\begin{aligned}
	&A_c=\begin{bmatrix}0_{n\times n} & I_n\\ 0_{n \times n}&0_{n \times n}\end{bmatrix},~\quad B_c=\begin{bmatrix}0_{n\times n}\\\frac{1}{m}I_n\end{bmatrix}\,.
	\end{aligned}
	\end{equation*}
We assume that the leading vehicle in the formation knows its own absolute position and velocity, while all the other vehicles can measure their relative distance to the preceding vehicle and their own absolute velocity. Please see Figure~\ref{fig:porsche}~\footnote{Image from \url{http://carsinamerica.net/porsche-918}}. Matrix $C$  defines the outputs of the system according to this information structure.
	\begin{figure}[thb]
	      \centering
	      	\includegraphics[width=0.5\textwidth]{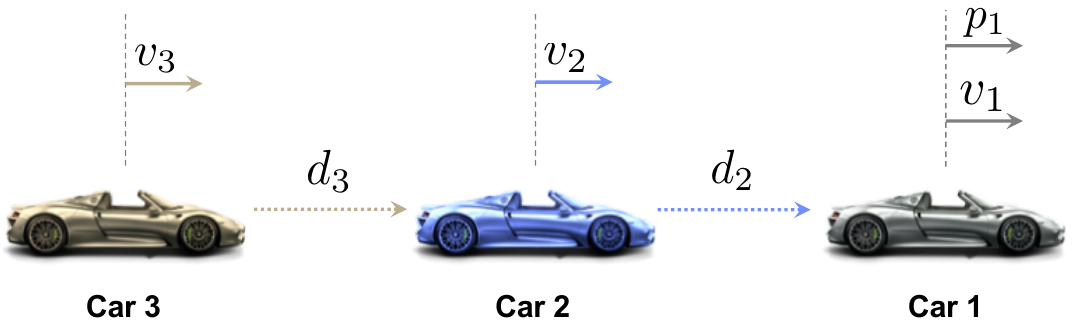}
	      	\setlength{\belowcaptionskip}{-10pt}
	      	\caption{}
	      		      	      		\label{fig:porsche}

	\end{figure}
A first order approximation of the continuous time model with sampling time $T_s=0.2\si{\second}$ (Euler discretization) is obtained as
\begin{equation*}
\begin{aligned}
&x_{k+1}=Ax_k+Bu_k+w_k\,, \quad y_k=Cx_k+w_k\,,\\
\end{aligned}
\end{equation*} 	
where $x_k \in \mathbb{R}^{2n}$, $u_k \in \mathbb{R}^n$, $y_k \in \mathbb{R}^{2n}$, $w_k \in \mathcal{W}$ for every $k$, $A=I_{2n}+A_cT_s$, $B=\left(I_{2n}+\frac{A_cT_s}{2}\right)B_cT_s$. Every vehicle can measure its own two outputs as per Figure~\ref{fig:porsche}. Hence,  the sparsity subspace is defined as $\mathbfcal{S}=\text{Sparse}(\mathbf{S})$, where each non-zero block of $\mathbf{S}$ according to (\ref{eq:L}) is $S=I_n \otimes \begin{bmatrix}1&1\end{bmatrix}$. 
\vspace{-0.1cm}
\subsection{Setting up the simulation}
We set a horizon of $N=15$ time steps. We consider the platooning of $8$ ``Porsche 918 Spyder'' cars  with mass $m=1700 \si{\kilogram}$. These cars have the capability of accelerating from $0\si{\metre \per \second}$ to $27 \si{\metre \per \second}$ in $2.2 \si{\second}$. Hence, we constrain the absolute value of the force input not to exceed $20\si{\kilo \newton}$.

 Disturbances for positions and velocities are taken from the set $\mathcal{W}\subseteq \mathbb{R}^2$, defined as the square  centered at the origin with edges $0.2\si{\metre}$, $\si{\metre \per \second}$ parallel to the axes. Hence, we assume that disturbances up to $0.1\si{\metre}$ on positions and up to $0.1 \si{\metre \per \second}$ on velocities are present at each time because of sliding effects and possible unevenness of the road.  We require that the vehicles maintain a safety distance of $2\si{\metre}$ between each other at every time. 

The cost function $J(\cdot)$ is chosen as a convex quadratic function of the states penalizing both the distance of each vehicle from an assigned target position and their velocities. Without any assumptions on statistics about the disturbances, $J(\cdot)$  is computed over the disturbance-free trajectory of the states. The distances from the assigned target positions are penalized with a weight of $0.01$ from time $0$ to time $N-1$  and with a weight of $0.2$ at time $N$. The velocities are penalized with a weight of $0.002$ from time $0$ to time $N-1$  and with a weight of $0.2$ at time $N$. It can be easily verified that $\mathbfcal{S}$ is not QI with respect to $\mathbf{CB}$ because $SCABS \not \in \text{Sparse}(S)$. Hence, there are no known approaches to obtain an exact convex representation of Problem~1. Let $J^\star$ be the optimal value of Problem~1. In what follows, we apply our techniques to derive upper and lower bounds on  $J^\star$. Furthermore, we compare the performance of our feasible controller with the one obtained using techniques from \cite{rotkowitz2012nearest} and show that strictly better performance can be obtained.
 \subsection{Performance bounds beyond QI sparsity constraints} The work in \cite{rotkowitz2012nearest} considered upper bounds on the minimum cost of Problem~1 based on determining QI sparsity subspaces that are subsets of $\mathbfcal{S}$. Determining  ``nearest'' QI sparsity subspaces which are subsets of $\mathbfcal{S}$ was shown to be intractable \cite{rotkowitz2012nearest}. For the considered example, it can be  verified by direct inspection that there is no binary matrix $T'\leq S$ with cardinality larger or equal than $10$ such that $\text{Sparse}(T')$ is QI with respect to $CA^kB$ for any integer $k \geq 0$.  Let us then consider the binary matrix $T<S$ with cardinality $9$:
 \begin{equation*}
T=\text{blkdiag}\left(I_3 \otimes \begin{bmatrix}
 1&1&0&0\\0&0&0&0
 \end{bmatrix}, \begin{bmatrix}1&1&0&0\\0&0&0&1\end{bmatrix}\right)\,.
 \end{equation*}
 
It can be verified that  $\text{Sparse}(T)$ is QI with respect to $CA^kB$ for any integer $k \geq 0$. Furthermore, $T$ is as close as possible to $S$ in the sense that its cardinality is maximized while still preserving QI. Let us now define the corresponding stacked operators. Let $\mathbf{T}\in \{0,1\}^{m(N+1)\times p(N+1)}$ be such that  its non-zero blocks as per (\ref{eq:L}) are all equal to $T$.   By \cite[Theorem 3]{CDC2017} we have that $\text{Sparse}(\mathbf{T})$ is QI with respect to $\mathbf{CB}$. Since $\text{Sparse}(T)$ is QI with respect to $CA^kB$ for every $k$ and as close as possible to $\text{Sparse}(S)$, then $\text{Sparse}(\mathbf{T})$ is also a QI sparsity subspace which is close  to $\mathbfcal{S}$. By solving Problem~2 with $\mathbfcal{R}=\text{Sparse}(\mathbf{T})$ the upper bound $\bar{J}_{\text{QI}}=131.3381$ on $J^\star$ was obtained.
 
We then solved Problem~2  with the GSS $\mathbfcal{R}= \mathbfcal{R}_{\mathbf{CB}}(\mathbf{S},\mathbf{Y}_{\text{max}}^{\mathbf{S}})$ and obtained the upper bound ${\bar{J}_{\text{GSS}}=123.7376}$ on $J^\star$. The feasible controller we obtained is thus $1-\frac{\bar{J}_{\text{GSS}}}{\bar{J}_{\text{QI}}}=5.8\%$ more performing. This improvement was possible thanks to the fact that sparsity preserving GSS's are significantly more general than QI sparsity subspaces that are subsets of $\mathbfcal{S}$. We also remark that $\mathbfcal{R}_{\mathbf{CB}}(\mathbf{S},\mathbf{Y}_{\text{max}}^{\mathbf{S}})$ can be  easily computed according to the procedure in Proposition~\ref{pr:cannot_be_larger}, whereas an algorithm to determine a QI sparsity subspace which is the nearest subset of $\mathbfcal{S}$ in polynomial time is not known \cite{rotkowitz2012nearest}.

In order to evaluate the precision of the upper bounds on $J^\star$ we computed a lower bound on $J^\star$ based on assuming that additional delayed communication is available as proposed in our work \cite{furieri2017value}.
%
It can be verified that a QI subspace which is a superset of $\mathbfcal{S}$ is obtained if vehicles  propagate the information to the single vehicle following it. Using this information relaxation we obtained the lower bound ${\underline{J}=119.6777}$ on $J^\star$. We can thus state that the feasible control policy computed using our approach is $\frac{\bar{J}_{\text{GSS}}-\underline{J}}{\underline{J}}=3.4\%$ suboptimal in the worst case where $J^\star=\underline{J}$.
  ~ All the instances of Problem~2 that we considered were solved with MOSEK \cite{andersen2000mosek}, called through MATLAB \cite{MATLAB} via YALMIP \cite{YALMIP}, on a computer equipped with a 16GB RAM and a 4.2 GHz quad-core Intel i7 processor. Computing the upper bound $\bar{J}_{\text{GSS}}$ required $2.1874\si{\second}$ of solver time. The  simulation results are shown in Figure \ref{fig:masses}.

	\begin{figure}[thb]
	      \centering
	      	\includegraphics[width=0.6\textwidth]{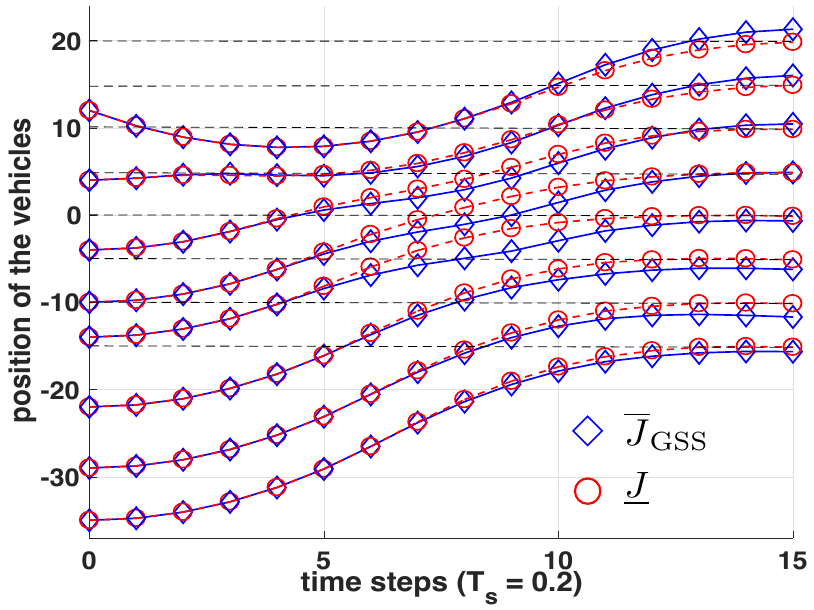}
	      		      	\setlength{\belowcaptionskip}{-13pt}
	      	\caption{ \small Trajectories of the vehicles' positions corresponding to our upper bound $\bar{J}_{\text{GSS}}$ (blue diamonds) and the infeasible lower bound obtained by restoring QI with techniques from \cite{furieri2017value} (red circles). The dashed straight lines represent the target positions. }
	      		\label{fig:masses}
	\end{figure}	
\section{Conclusions}
We proposed a novel technique to compute robust controllers which comply with any information structure. Specifically, the information structure  could be non-QI. 
 Our approach is based on the notion of Generalized Sparsity Subspaces (GSS). We derived conditions so that if the disturbance feedback parameter lies in a GSS the sparsity of the original output feedback controller is preserved and derived a procedure to improve the choice of a GSS for maximized performance guarantees. Furthermore, we showed that the upper bounds we propose are at least as low as the ones obtained with the techniques in \cite{rotkowitz2012nearest} and can be even lower in general. We showed performance improvement with respect to \cite{rotkowitz2012nearest} on a platooning example.
An immediate future development is investigating heuristics to tailor the choice of a GSS to any specific dynamical system.  It is also relevant to determine other classes of sparsity preserving subspaces as per Figure~\ref{fig:mapping} which cannot be expressed as a GSS. 

	\allowdisplaybreaks
	\section*{Appendix: Mathematical Notation}
	\label{app:notation}
	We define the following matrices and vectors. 
		\begin{equation}
	\label{eq:L}
	\begin{aligned}
	&\mathbf{L}=\hspace{-0.1cm}
		\begingroup 
\setlength\arraycolsep{1.4pt}
\begin{bmatrix}
	L_{0,0}&0_{m \times p}&\cdots&0_{m \times p}\\
	\vdots&\vdots&\ddots&\vdots\\
	L_{N\text{-}1,0}&\cdots&L_{N\text{-}1,N\text{-}1}&0_{m \times p}\\
	0_{m \times p}&\cdots&0_{m \times p}&0_{m \times p}
	\end{bmatrix}\hspace{-0.125cm},~\mathbf{g}=\hspace{-0.1cm}\begin{bmatrix}g_0\\\vdots\\g_{N\text{-}1}\\0_{m \times 1}\end{bmatrix}\endgroup\hspace{-0.1cm}.
	\end{aligned}
	\end{equation}
	The matrix blocks above are $L_{k,j}\in \mathbb{R}^{m \times p}$, $g_i \in \mathbb{R}^{m}$ as in (\ref{eq:affine_feedback_constraints}), and the $0_{m \times p}$ blocks enforce causality of the controller.
	\begin{align*}
	&\mathbf{A}=\begin{bmatrix}I_n&A^\mathsf{T}&\cdots&{A^N}^\mathsf{T}\end{bmatrix}^\mathsf{T}\in \mathbb{R}^{n(N+1)\times n}\,,\\
	&\mathbf{E}=\hspace{-0.1cm}
		\begingroup 
\setlength\arraycolsep{3pt}	
	\begin{bmatrix}0_{n \times n}&0_{n \times n}&\cdots&0_{n \times n}&0_{n \times n}\\
	I_n&0_{n \times n}&\cdots&0_{n \times n}&0_{n \times n}\\
	A&I_n&\cdots&0_{n \times n}&0_{n \times n}\\
	\vdots&\vdots&\ddots&\vdots&\vdots\\
	A^{N-1}&A^{N-2}&\cdots&I_n&0_{n \times n}
	\end{bmatrix}\endgroup\hspace{-0.1cm}\in  \mathbb{R}^{n(N\text{+}1)\times n(N\text{+}1)},\\
	&\mathbf{B}=\mathbf{E}(I_{N+1} \otimes B)\,, \quad  	\mathbf{E}_D=\mathbf{E}(I_{N+1} \otimes D)\,,\\
	&\mathbf{C}=I_{N+1} \otimes C\,, \quad \mathbf{H}=I_{N+1} \otimes H \,,\\
	&\mathbf{U}=\begin{bmatrix}
	I_N\otimes U & 0_{Ns \times n}\\0_{r \times nN}&R
	\end{bmatrix}\,, \quad \mathbf{V}=\begin{bmatrix}
	I_{N}\otimes V&0_{Ns \times m} \\ 
	0_{r \times mN}&0_{r \times m}
	\end{bmatrix}\,,\\
	&F=\mathbf{UB+V}\,, \quad G=\mathbf{UE}_D\,,\quad c=\begin{bmatrix} \mathbf{1}_N\otimes b\\z\end{bmatrix}-\mathbf{UA}x_0 \,.\\
	\end{align*}
	\bibliographystyle{IEEEtran}
	\bibliography{IEEEabrv,references}

	\end{document}